\documentclass[reqno]{amsart}

\usepackage{booktabs}
\usepackage{IEEEtrantools}
\usepackage{amssymb,latexsym,amsfonts,amsmath}
\usepackage{mathrsfs}
\usepackage{graphicx}
\usepackage{dsfont}
\usepackage{tikz}
\usetikzlibrary{calc,shapes,arrows}

\usepackage{xspace}

\newtheorem{theorem}{Theorem}[section]
\newtheorem{lemma}[theorem]{Lemma}

\newtheorem{problem}[theorem]{Problem}

\newtheorem{corollary}[theorem]{Corollary}

\newtheorem{definition}[theorem]{Definition}

\newtheorem{remark}[theorem]{Remark}

\numberwithin{equation}{section}

\newcommand{\reals}{{\mathbb{R}}}
\newcommand{\nats}{{\mathbb{N}}}
\newcommand{\identity}{{\mathbb{I}}}

\usepackage{algorithm}
\usepackage{algpseudocode}

\usepackage{tcolorbox}

\newtcolorbox{resp}[1][]{%
	colback=gray!5!white,%
	colframe=gray!80!black,%
	size=small,%
	boxrule=1pt,%
	halign title=flush center,%
	coltitle=black,%
	#1%
}

\usepackage{fancyhdr}

\newenvironment{nouppercase}{%
	\renewcommand{\uppercasenonmath}[1]{}}{}

\begin{document}

\begin{abstract}
This work is concerned with synthesizing safety controllers for discrete-time nonlinear systems beyond polynomials with \emph{unknown} mathematical models using the notion of $k$-inductive control barrier certificates ($k$-CBCs). Conventional CBC conditions (with $k=1$) for ensuring safety over dynamical systems are often restrictive, as they require the CBCs to be non-increasing at every time step. Inspired by the success of $k$-induction in software verification, $k$-CBCs relax this requirement by allowing the barrier function to be non-increasing over $k$ steps, while permitting $k-1$ (one-step) increases, each up to a threshold $\epsilon$. This relaxation enhances the likelihood of finding feasible $k$-CBCs while providing safety guarantees across the dynamical systems. Despite showing promise, existing approaches for constructing $k$-CBCs often rely on \emph{precise mathematical knowledge} of system dynamics, which is frequently unavailable in practical scenarios. In this work, we address the case where the underlying dynamics are unknown—a common occurrence in real-world applications—and employ the concept of \emph{persistency of excitation}, grounded in Willems \textit{et al.}'s fundamental lemma. This result implies that input-output data from a single trajectory can capture the behavior of an unknown system, provided the collected data fulfills a specific rank condition. We employ sum-of-squares (SOS) programming to synthesize the $k$-CBC as well as  the safety controller directly from data while ensuring the safe behavior of the unknown system. The efficacy of our approach is demonstrated through a set of physical benchmarks with unknown dynamics, including a DC motor, an RLC circuit, a nonlinear \emph{nonpolynomial} car, and a nonlinear polynomial Lorenz attractor.\\

\end{abstract}

\title{{\LARGE Learning $k$-Inductive Control Barrier Certificates for Unknown  \vspace{0.2cm}\\ Nonlinear Dynamics Beyond Polynomials\vspace{0.4cm}}}

\author{{\bf {\large Ben Wooding \and Abolfazl Lavaei\vspace{0.4cm}}}\\{\normalfont School of Computing, Newcastle University, United Kingdom}\\ \texttt{\{ben.wooding,abolfazl.lavaei\}@newcastle.ac.uk}}

\pagestyle{fancy}
\lhead{}
\rhead{}
  \fancyhead[OL]{Ben Wooding and Abolfazl Lavaei}

  \fancyhead[EL]{Learning $k$-Inductive Control Barrier Certificates for Unknown Nonlinear Dynamics Beyond Polynomials}
  \rhead{\thepage}
 \cfoot{}
 
\begin{nouppercase}
	\maketitle
\end{nouppercase}

\section{Introduction}
\label{sec:introduction}

Over the past two decades, formal methods have gained prominence in \emph{verifying and synthesizing controllers} for dynamical systems to ensure safety properties, preventing the system from evolving into unsafe regions. These approaches are particularly crucial in \emph{safety-critical} applications~\cite{mcgregor2017analysis}, where undesired behavior can lead to severe consequences, including loss of life, injury, or significant financial losses. While \emph{formal verification} focuses on rigorously checking whether a system meets the desired specifications, the \emph{synthesis problem} involves designing a controller—typically a state-feedback architecture—for dynamical models with control inputs to enforce the property of interest.

Control barrier certificates (CBCs)~\cite{prajna2004safety,ames2019control,xiao2023safe,wooding2024protect} offer a \emph{discretization-free} method to provide formal guarantees on system behavior. Similar to Lyapunov functions, CBCs are defined over the system's state space and satisfy specific inequalities related to both the function itself and the system dynamics. An appropriate level set of a CBC can separate unsafe regions from system trajectories starting from a set of initial conditions. The existence of such a function provides formal certification, along with a controller that enforces system safety. While primarily designed to ensure safety, these methods can also be applied to other properties like reachability~\cite{prajna2007convex}. A graphical illustration of a CBC is provided in Fig.~\ref{fig:BC-demo}.

\begin{figure}[t]
	\centering
\includegraphics[width=0.43\columnwidth]{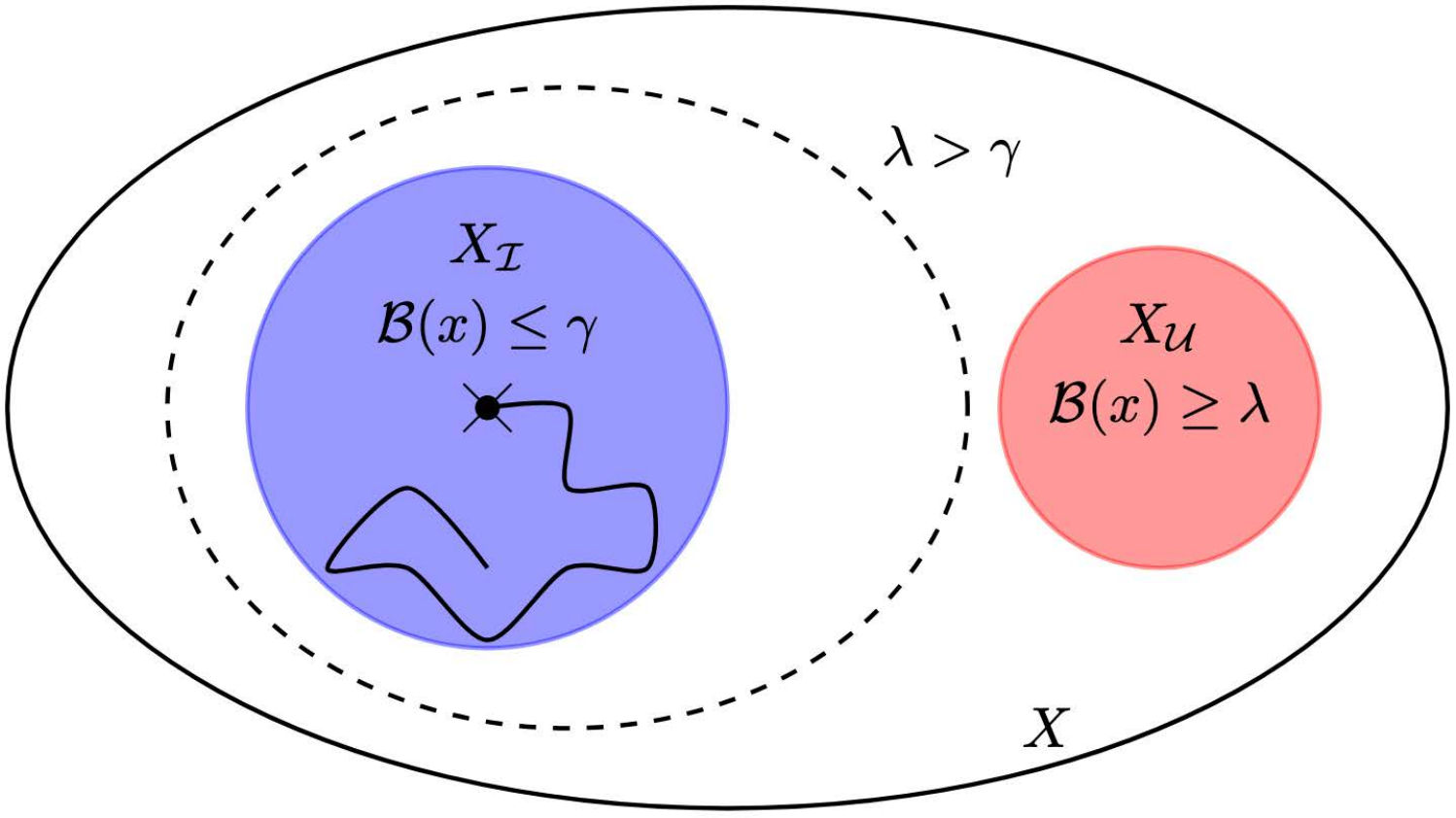}
	\caption{A CBC $\mathcal{B}(x)$ for a dynamical system, with the dashed line representing the initial level set $\mathcal{B}(x) = \gamma$. Purple and red circles are initial and unsafe regions, respectively.}
	\label{fig:BC-demo}
\end{figure}

\noindent\textbf{$k$-Inductive CBCs.} While CBCs show great promise for offering safety certificates, they are restrictive in the sense that they require the CBC to be non-increasing at every time step to certify safety. Since CBC approaches are only \emph{sufficient}, if their inherent conservatism leads to no CBC being found, no conclusions can be drawn about the system safety. To address this non-increasing limitation, $k$-inductive control barrier certificates ($k$-CBCs) have been introduced in the relevant literature~\cite{Anand2021kSafety,anand2022k}, which relax the strict non-increasing condition by allowing $k-1$ one-step increases, each up to a threshold $\epsilon$. This relaxation enhances the likelihood of finding feasible $k$-CBCs and potentially providing safety guarantees for systems where standard CBCs may fail.

While $k$-CBCs offer a more flexible approach to providing safety guarantees compared to conventional CBCs, they rely on \emph{precise knowledge} of the system's mathematical model. However, in real-world scenarios, such models are often unavailable or too complex to be constructed accurately. 
This highlights the need for alternative \emph{data-driven approaches} that use only \emph{input-output data} to ensure safety, bypassing the requirement for an explicit model while offering formal correctness guarantees.

\noindent\textbf{Data-Driven Techniques.}  Two promising data-driven approaches in the literature provide formal analysis of unknown systems using collected data. The first is the \emph{scenario approach}~\cite{calafiore2006scenario,campi2009scenario}, which initially solves the problem using data and then translates the results back to unknown models through intermediate steps involving \emph{chance constraints}~\cite{esfahani2014performance}. While scenario approaches show great potential for offering formal guarantees over unknown systems, they require the data to be \emph{independent and identically distributed (i.i.d.)}. This implies that only one input-output data pair is allowed to be obtained from each trajectory~\cite{calafiore2006scenario}, necessitating the collection of \emph{multiple independent trajectories} to achieve a desired confidence level, based on a known closed-form relationship between them. As a result, the scenario approach is especially suited for system simulators, where generating numerous independent trajectories is feasible.

An alternative data-driven technique that complements the scenario approach is the \emph{non-i.i.d. trajectory-based approach}. Unlike the scenario approach, this method only requires a \emph{single input-output trajectory} from the unknown system, collected over a \emph{specific time horizon}~\cite{de2019formulas}. It leverages the concept of \emph{persistent excitation}, implying that the trajectory must satisfy a rank condition for certain classes of systems to sufficiently excite the system's dynamics, as introduced by Willems \textit{et al.}'s fundamental lemma~\cite{willems2005note}. When a trajectory is \emph{persistently excited}, it contains enough information about the system's behavior to enable analysis of the unknown system. This approach is particularly beneficial since it reduces the need for multiple independent trajectories, making it practical for situations where collecting several distinct trajectories is not feasible.

\noindent\textbf{Key Contributions.} The central contribution of this work is, for the first time, the construction of $k$-inductive control barrier certificates using only a single input-state trajectory from \emph{unknown} discrete-time nonlinear systems beyond polynomials, ensuring safe system behavior.  We formulate the required conditions based solely on data for designing $k$-CBCs and synthesizing their safety controllers, demonstrating that these conditions generalize the standard CBC framework (\emph{i.e.,} $k = 1$). To achieve this, we develop a sum-of-squares (SOS) programming framework for our design approach using data from a single trajectory, ensuring the safe behavior of the unknown system. Our method is particularly valuable as it enables a systematic approach for handling even \emph{nonlinear nonpolynomial} dynamics, while $k$-CBCs are derived through SOS programming. We demonstrate the effectiveness of our method through a series of physical case studies, including a DC motor, an RLC circuit, a nonlinear \emph{nonpolynomial} car, and a nonlinear polynomial Lorenz attractor, all with unknown dynamics, illustrating its potential for ensuring system safety in diverse scenarios.

\noindent\textbf{Related Works.} A body of research has explored data-driven analysis of unknown systems using the \emph{scenario approach}. Existing results include the construction of either barrier certificates~\cite{nejati2023formal,lindemann2021learning,nejati2023data} or finite abstractions~\cite{kazemi2024data,lavaei2022data,devonport2021symbolic,ajeleye2023data,coppola2024data,banse2023data2,lavaei2023symbolic,samari2024data} using scenario approach for use in formal verification and controller synthesis. There is also a rich body of literature on data-driven methods using the \emph{single-trajectory approach}. Existing results include ensuring stability and invariance properties for unknown systems~\cite{de2019formulas,dai2023data,de2019formulas,Mahdieh2024singletrajectory}, constructing finite abstractions~\cite{samari2024Abstraction} as well as constructing CBCs for both continuous- and discrete-time systems~\cite{nejati2022data,wang2024convex,samari2024singletrajectory}. Other data-driven methods for constructing barrier certificates include the sublinear algorithm~\cite{han2015sublinear} and rapidly-exploring random trees~\cite{ahmad2022adaptive}.

In the literature on \emph{model-based} $k$-inductive barrier certificates, these methods have been proposed for both deterministic and stochastic systems, as explored in~\cite{Anand2021kSafety,anand2022k}. The work in~\cite{ren2024formal} leverages neural networks to find candidate $k$-inductive CBCs for systems with known  dynamics, offering an efficient approach to controller synthesis. For \emph{continuous-time} model-based systems, a related concept, $t$-barriers, has been introduced in~\cite{bak2018t}, expanding the application of barrier certificates to continuous-time dynamics and providing formal guarantees for safety across various time intervals.

To the best of our knowledge, only one study has explored \emph{data-driven} methods with $k$-induction thus far~\cite{murali2022scenario}. However, this study focuses exclusively on verification using $k$-inductive barrier certificates, \emph{without addressing controller synthesis}, and relies on the \emph{scenario approach}, which requires multiple i.i.d. samples. For instance, this approach required approximately 60,000 samples for a 2-dimensional case study, as reported in~\cite{murali2022scenario}. 
In contrast, our method requires only a single persistently excited trajectory with a horizon of at least $3$ time steps for the same example. Additionally, the scenario approach in~\cite{murali2022scenario} requires knowledge of the \emph{system's Lipschitz constant}, which may not be available for unknown systems, leading to further complexity as this constant must be estimated from the data. Our trajectory-based approach, in contrast, eliminates the need for such a Lipschitz constant, but assumes an extensive dictionary of nonlinear functions. It is worth mentioning that the scenario approach in~\cite{murali2022scenario} does not impose any restrictions on the form of the $k$-CBC, whereas in this work, we require it to be quadratic.

\noindent\textbf{Organization.} 
The rest of the paper is organized as follows: In Section~\ref{sec:problem}, we define the systems under analysis and outline the problem we aim to solve. Section~\ref{sec:linear-data} introduces our data-driven approach for discrete-time \emph{linear} systems, where we present our theorem for $k$-CBCs ensuring safety for this class of models. In Section~\ref{sec:nonlinear-data}, we extend these results to \emph{nonlinear systems}. Section~\ref{sec:SOS} details the sum-of-squares (SOS) formulations used to solve the $k$-CBC conditions. In Section~\ref{sec:case-studies}, we provide multiple physical case studies with unknown dynamics to demonstrate the effectiveness of our data-driven approach. Finally, Section~\ref{sec:conclusion} concludes the paper and discusses directions for future research.

\noindent\textbf{Notation.} We use $\reals,$ $\reals_{\geq 0}$, and $\reals_{>0}$ to denote the real numbers, non-negative  and positive real numbers, respectively. We utilize $\nats$ and $\nats_{>0}$ for the non-negative and positive integers, respectively. We denote $k$ iterations of system's dynamics by $x^{k+}$. Similarly, for a square matrix, we use $A^k$ to represent $k$ repeated multiplications of the matrix, \emph{e.g.,} when $k = 3$, then $A^3 = A A A$. A \emph{symmetric} matrix $P\in\reals^{n\times n}$ is said to be positive definite, if all its eigenvalues are positive. We define an $(n \times n)$ identity matrix by  $\identity_n$, while an $(n \times m)$ zero matrix is denoted by $\mathbf{0}_{n\times m}$. The horizontal concatenation of vectors $x_i \in \mathbb R^n$ into an $n \times N$ matrix is written as $\begin{bmatrix} x_1 & x_2 & \dots & x_N \end{bmatrix}$\!.

\section{Problem Formulation}\label{sec:problem}

To ensure a cohesive presentation and for the sake of readability, we first develop our data-driven framework for discrete-time \emph{linear} systems and then generalize it to \emph{nonlinear} systems beyond polynomials.

\begin{definition}
\label{def:linear-system-description}
    A discrete-time linear system (dt-LS) is described by
    \begin{equation}
    \label{eq:dt-LS}
        \Sigma\!: x^+ = Ax+Bu,    
    \end{equation}
    where $A\in\reals^{n\times n}$, $B\in\reals^{n\times m}$, $x\in X$ is the system state, $u\in U$ is the control input, with $X\subset\reals^n$, and $U\subset\reals^m$ being the state and input sets, respectively, while $x^+$ denotes the system's state at the next time step, i.e., $x^+ = x(k+1), k\in\mathbb N$. We assume that both matrices $A$ and $B$ are unknown in our setting.
\end{definition}

The following definition, borrowed from \cite{prajna2004safety}, introduces control barrier certificates, which are used to ensure safety specifications for discrete-time systems.

\begin{definition}[CBC -~\cite{prajna2004safety}]\label{def:safe_barrier-conditions} 
Consider a dt-LS $\Sigma$ in Definition~\ref{def:linear-system-description} with $X_{\mathcal I},X_{\mathcal U}\subset X$ being its initial and unsafe sets, respectively. A function $\mathcal{B}: X\rightarrow \reals_{\geq0}$ is  a control barrier certificate (CBC) for $\Sigma$ if there exist $\gamma,\lambda\in\reals_{\geq0}$, with $\lambda>\gamma$, such that
\begin{subequations}
    \begin{align}
    \label{eq:safe_CBC_cond1}
        \mathcal{B}(x) \leq \gamma, &\quad\quad\quad\forall x\in X_{\mathcal I}, \\
    \label{eq:safe_CBC_cond2}
    \mathcal{B}(x) \geq \lambda, &\quad\quad\quad\forall x\in X_{\mathcal U}, \end{align}
    and $ \forall x\in X, \exists u\in U$\!, 
    \begin{equation}
    \label{eq:safe_CBC_cond3}
    \mathcal{B}(x^+) \leq \mathcal{B}(x).
    \end{equation}
    \end{subequations}
\end{definition}

We denote by $x_{x_0u}(k)$ the state trajectory of $\Sigma$ at time $k\in\nats_{\geq 0}$ under the input signal $u(\cdot)$ starting from an initial condition $x_0=x(0)$.

As evident, condition~\eqref{eq:safe_CBC_cond3} requires the CBC to be non-increasing at every time step, which imposes a restriction. We now present the notion of a $k$-inductive CBC~\cite{Anand2021kSafety} as a relaxation of Definition~\ref{def:safe_barrier-conditions}.

\begin{definition}[$k$-CBC - \cite{Anand2021kSafety}]
    \label{def:safe_k-inductive-CBC}
    Given a dt-LS $\Sigma$ in Definition~\ref{def:linear-system-description}, a function $\mathcal{B}: X\rightarrow \reals_{\geq0}$ is  a $k$-inductive control barrier certificate ($k$-CBC) for $\Sigma$ if there exist $k\in\nats_{>0}$, $\gamma, \lambda, \epsilon\in\reals_{\geq0}$, with $\lambda>\gamma + (k-1)\epsilon$, such that
    \begin{subequations}
    \begin{align}
    \label{eq:safe_kCBC_cond1}
        \mathcal{B}(x) \leq \gamma, &\quad\quad\quad\forall x\in X_{\mathcal I}, \\
    \label{eq:safe_kCBC_cond2}
    \mathcal{B}(x) \geq \lambda, &\quad\quad\quad\forall x\in X_{\mathcal U}, \end{align}
    and $ \forall x\in X, \exists u\in U$,
    \begin{align}
    \label{eq:safe_kCBC_cond3}
    \mathcal{B}(x^+) &\leq \mathcal{B}(x) + \epsilon, \\
    \label{eq:safe_kCBC_cond4}
    \mathcal{B}(x^{k+}) &\leq \mathcal{B}(x).
    \end{align}
    \end{subequations}
\end{definition}

Notice that~\eqref{eq:safe_kCBC_cond3} is a relaxed version of~\eqref{eq:safe_CBC_cond3}, allowing a one-step increase of up to a threshold $\epsilon$, while condition~\eqref{eq:safe_kCBC_cond4} enforces condition~\eqref{eq:safe_CBC_cond3} but after $k$ time steps. The constraint $\lambda>\gamma + (k-1)\epsilon$ is crucial to avoid excessive relaxation of the new condition~\eqref{eq:safe_kCBC_cond3}. When $\epsilon$ equals zero and $k$ equals one, the original CBC conditions with $\lambda>\gamma$ is recovered.

The following theorem, borrowed from~\cite[Theorem 3.2]{Anand2021kSafety}, provides safety guarantees for discrete-time dynamical systems by leveraging the notion of $k$-inductive CBC from Definition~\ref{def:safe_k-inductive-CBC}.

\begin{theorem}[\cite{Anand2021kSafety}]
    Given a dt-LS $\Sigma$, suppose $\mathcal{B}$ is a $k$-CBC for $\Sigma$ as in Definition~\ref{def:safe_k-inductive-CBC}. Then, the state trajectory $x_{x_0u}(k)$ remains safe (i.e., $x_{x_0u}(k)\notin X_{\mathcal U}$) for any $x_0\in X_{\mathcal I}$ and $k\in\nats_{\geq 0}$, with a control input satisfying conditions~\eqref{eq:safe_kCBC_cond3} and~\eqref{eq:safe_kCBC_cond4}.
\end{theorem}

While the concept of $k$-CBC offers great potential for providing safety guarantees in dynamical systems with more relaxed conditions compared to conventional CBC, satisfying conditions \eqref{eq:safe_kCBC_cond3} and \eqref{eq:safe_kCBC_cond4} is not feasible due to the unknown dynamics involved in their left-hand side. This highlights the importance of developing \emph{data-driven} approaches to design such $k$-CBCs and their corresponding safety controllers for unknown systems based solely on input-state data.

We now summarize the problem we aim to address in this work.

\begin{resp}
\begin{problem}\label{Problem}
        Consider the discrete-time dynamical system (either~\eqref{eq:dt-LS} or~\eqref{eq:dt-NS}) with unknown matrices $A$ and $B$, and initial and unsafe sets $X_{\mathcal I}, X_{\mathcal U}\subset X$. Design a $k$-CBC $\mathcal B$ along with a controller $u$ solely based on a single input-state trajectory from the unknown dynamics to ensure the safety of the system.
    \end{problem}
\end{resp}

In the next section, we offer our data-driven approach to address Problem~\ref{Problem} for dt-LS in \eqref{eq:dt-LS}.

\section{Data-Driven Design of $k$-CBC for dt-LS} \label{sec:linear-data}

In our data-driven setting, we collect input-state data from the unknown discrete-time dynamical system over the time horizon $[0,1,\ldots,T-1]$, where $T\in\nats_{>0}$ is the number of collected samples:
\begin{subequations}\label{data}
\begin{align}
\label{eq:input-data}
    &\mathcal{U}_{0} = [u(0),u(1),\ldots,u(T-1)]\in\reals^{m\times T},\\
    \label{eq:state-data}
    &\mathcal{X}_{0} = [x(0),x(1),\ldots,x(T-1)] \in\reals^{n\times T},\\
    \label{eq:output-data}
    &\mathcal{X}_{1} = [x(1),x(2),\ldots,x(T)] \in\reals^{n\times T}.
\end{align}
\end{subequations}

We call the collected data in \eqref{eq:input-data}-\eqref{eq:output-data} a \emph{single input-state trajectory}.
 
Since both $x^+$ and $x^{k+}$ appear in conditions \eqref{eq:safe_kCBC_cond3} and \eqref{eq:safe_kCBC_cond4}, we now present the following lemma, inspired from~\cite{de2019formulas}, to obtain a data-based representation of both closed-loop dt-LS in~\eqref{def:linear-system-description} and its evolution after $k$ time steps.

\begin{lemma}\label{lem:Q-matrix}
    Let matrix $Q$ be a $(T\times n)$ matrix such that
    \begin{equation}
    \label{eq:Q-matrix-X0}
        \identity_n = \mathcal{X}_{0}Q,
    \end{equation}
    with $\mathcal{X}_{0}$ being an $(n\times T)$ full row-rank matrix. If one designs $u=Kx = \mathcal{U}_{0}Qx$, then the closed-loop system $x^+=Ax+Bu$ has the following data-based representation:
    \begin{equation}  \label{eq:A+BK=X1TQ}
        x^+ = \mathcal{X}_{1}Qx, \quad\text{equivalently}\quad
        A+BK = \mathcal{X}_{1}Q.
    \end{equation}
    Additionally, the data-based evolution of the system after $k$ steps is \begin{equation}
    	\label{eq:inductive-equation}
    	x^{k+} = (\mathcal{X}_{1}Q)^kx.
    \end{equation}
\end{lemma}
\begin{proof}
Since $u=Kx = \mathcal{U}_{0}Qx$ and leveraging~\eqref{eq:Q-matrix-X0}, the closed-loop dt-LS can be written as
\begin{align*}
   x^+ \!=\!  Ax \!+\! B u \!=\! (A\!+\!BK)x = \begin{bmatrix}
        B & A
    \end{bmatrix} \!\!\begin{bmatrix}
        K \\ \identity_n
    \end{bmatrix} x \overset{\eqref{eq:Q-matrix-X0}}{=}
    \underbrace{\begin{bmatrix}
        B & A
    \end{bmatrix} \!\!\begin{bmatrix}
        \mathcal{U}_{0} \\ \mathcal{X}_{0}
    \end{bmatrix}}_{\mathcal{X}_{1} } Qx = \mathcal{X}_{1}Qx,
\end{align*}
with $\mathcal{U}_{0}$ and $\mathcal{X}_{0}$ as in~\eqref{eq:input-data} and~\eqref{eq:state-data}, respectively. Therefore, $x^+ = \mathcal{X}_{1}Qx$, equivalently, $A+BK = \mathcal{X}_{1}Q$, is the data-based representation of the closed loop dt-LS.
    
    \noindent Since $x^{k+}$ is obtained by applying $k$ iterations of $Ax+Bu$, it follows that
    \begin{align*}
    	x^+ &= (A+BK)x = \mathcal{X}_{1}Qx, \\
    	x^{2+} &= \mathcal{X}_{1}Qx^+ = (\mathcal{X}_{1}Q)^2x,
    \end{align*}
    and by induction, we have
    \begin{equation*}
    	x^{k+} = (\mathcal{X}_{1}Q)^kx,
    \end{equation*}
    thus completing the proof.
    \end{proof}

\begin{remark}
The required condition for the existence of the matrix $Q$ in \eqref{eq:Q-matrix-X0} is that $\mathcal{X}_{0}$ must be full row-rank. This condition also ensures that the collected data is persistently excited~\cite{willems2005note}. To fulfill this requirement, the number of data points $T$ must exceed $n$, i.e., $T > n$. Since $\mathcal{X}_0$ is constructed from data, this rank condition can be readily satisfied during the data collection process.
\end{remark}
    
We now define the structure of the $k$-CBC in the form $\mathcal{B}(x) = x^\top Px$, where $P \in \mathbb{R}^{n \times n}$ is a \emph{symmetric} positive-definite matrix. We leverage Lemma \ref{lem:Q-matrix} to present our main theorem in this section, which offers data-driven safety guarantees for unknown dt-LS using $k$-CBC.

\begin{theorem}
    \label{thm:safe_k-inductive-CBC}
    Consider an unknown dt-LS $x^+ = Ax+Bu$, with its data-based representations $x^+=\mathcal{X}_{1}Qx$ and $x^{k+} = (\mathcal{X}_{1}Q)^kx$. Suppose there exist a symmetric positive-definite matrix $P\in\reals^{n\times n}$, $k\in\nats_{>0}$, $\gamma, \lambda, \epsilon\in\reals_{\geq0}$, with $\lambda>\gamma + (k-1)\epsilon$, such that the following conditions are satisfied:
    \begin{subequations}
\begin{align}
    \label{eq:safe_kCBC_data1}
        x^\top Px &\leq \gamma, \quad\quad\quad\quad\:\forall x\in X_{\mathcal I}, \\
    \label{eq:safe_kCBC_data2}
    x^\top Px &\geq \lambda, \quad\quad\quad\quad\:\forall x\in X_{\mathcal U},\\
    \label{eq:safe_kCBC_data3}
    x^\top Q^\top \mathcal{X}_{1}^\top P\mathcal{X}_{1}Qx &\leq x^\top Px \!+\! \epsilon, \quad\forall x\in X,\\
    \label{eq:safe_kCBC_data4}
    ((\mathcal{X}_{1}Q)^k)^\top P(\mathcal{X}_{1}Q)^k &\leq P.
    \end{align}
    \end{subequations}
    Then $\mathcal{B}(x) = x^\top Px$ is a $k$-CBC and  $u = \mathcal{U}_{0}Qx$ is its corresponding safety controller for the unknown dt-LS.
\end{theorem}
\begin{proof}
    Since $\mathcal{B}(x) = x^\top Px$, it is straightforward that condition~\eqref{eq:safe_kCBC_data1}-\eqref{eq:safe_kCBC_data2} implies~\eqref{eq:safe_kCBC_cond1}-\eqref{eq:safe_kCBC_cond2}, respectively. We now proceed to show condition~\eqref{eq:safe_kCBC_cond3}. Considering the definition of $\mathcal{B}(x)$ and using the results of Lemma~\ref{lem:Q-matrix}, one has
    \begin{align*}
    \mathcal{B}(x^+)
    &= (Ax+Bu)^\top P(Ax+Bu) = ((A+BK)x)^\top P((A+BK)x) \\
    &= x^\top(A+BK)^\top P(A+BK)x\overset{\eqref{eq:A+BK=X1TQ}}{=}x^\top(\mathcal{X}_{1}Q)^\top P(\mathcal{X}_{1}Q)x=x^\top Q^\top\mathcal{X}_{1}^\top P\mathcal{X}_{1}Qx.
    \end{align*}
    By leveraging condition~\eqref{eq:safe_kCBC_data3}, one has 
    \begin{align*}
    	\mathcal{B}(x^+)
    	= x^\top Q^\top\mathcal{X}_{1}^\top P\mathcal{X}_{1}Qx \leq \underbrace{x^\top P x}_{\mathcal{B}(x)} + \epsilon.
    \end{align*}
    Lastly, we demonstrate condition~\eqref{eq:safe_kCBC_cond4}:
    \begin{align*}
         \mathcal{B}(x^{k+})  &\overset{\eqref{eq:inductive-equation}}{=} ((\mathcal{X}_{1}Q)^kx)^\top P((\mathcal{X}_{1}Q)^kx)\\
         &= x^\top ((\mathcal{X}_{1}Q)^k)^\top P((\mathcal{X}_{1}Q)^k)x\overset{\eqref{eq:safe_kCBC_data4}}{\leq}
          \underbrace{x^\top P x}_{\mathcal{B}(x)}.
    \end{align*}
     Consequently, $\mathcal{B}(x) = x^\top Px$
    is a $k$-CBC and $u=Kx=\mathcal{U}_{0}Qx$ is its corresponding safety controller for the unknown dt-LS, thus completing the proof.
\end{proof}

We present Algorithm~\ref{alg:linear} that provides simplified steps for finding the $k$-CBC and its safety controller for dt-LS. It is worth highlighting that since the matrix $Q$ is first computed from \eqref{eq:Q-matrix-X0} following Step 2 of Algorithm~\ref{alg:linear}, there is no bilinearity in conditions \eqref{eq:safe_kCBC_data3}-\eqref{eq:safe_kCBC_data4}, as they are solved solely for $P$ and $\epsilon$.

\begin{algorithm}[t]
		\caption{Design of $k$-CBC and safety controller for dt-LS}\label{alg:linear}
		\begin{algorithmic}[1]
			\Require Regions of interest $X,X_{\mathcal I},X_{\mathcal U}$, relaxation horizon $k$
			\State Collect trajectories $\mathcal{U}_0, \mathcal{X}_0,\mathcal{X}_1$, according to~\eqref{data}, with $\mathcal{X}_0$ being full row-rank
			\State Solve~\eqref{eq:Q-matrix-X0} for $Q$
			\If{Symmetric positive-definite $P$ and $\lambda,\gamma,\epsilon$ with $\lambda>\gamma + (k-1)\epsilon$ are found by satisfying~\eqref{eq:safe_kCBC_data1}-\eqref{eq:safe_kCBC_data4}}
			\Statex \textbf{Return:} $k$-CBC $\mathcal{B}(x) = x^\top Px$ and safety controller $u = \mathcal{U}_{0}Qx$
			\Else
			\Statex \textbf{Return:} {Feasible solution not found with given $k$}
			\EndIf
		\end{algorithmic}
\end{algorithm}

In the next section, we generalize our data-driven approach  to accommodate discrete-time \emph{nonlinear systems}.

\section{Discrete-Time Nonlinear Systems}
\label{sec:nonlinear-data}

\begin{definition}
	\label{def:nonlinear-system-description}
	A discrete-time nonlinear system (dt-NS) is described by
	\begin{equation}
		\label{eq:dt-NS}
		\Sigma: x^+ = A{M}(x)+Bu,    
	\end{equation}
	where $A\in\reals^{n\times N}$, $B\in\reals^{n\times m}$, and $u\in U$ is a control input, with $U\subset\reals^m$ being its input set. In addition, ${M}(x)\in\reals^N$ is a transition map partitioned into linear terms in states $x\in \reals^{n}$ and nonlinear terms in states $\mathcal Z(x)\in \reals^{N-n}$, structured as
    \begin{equation*}
        M(x)= \begin{bmatrix}
        	x\vspace{-0.25cm}\\
        	\tikz\draw [thin,dashed] (0,0) -- (1.25,0);\\
        	\mathcal Z(x)
        \end{bmatrix}\!\!.
    \end{equation*}
\end{definition}

In our work, both matrices $A$ and $B$ are treated as \emph{unknown}, reflecting real-world scenarios. While the exact form of ${M}(x)$ is assumed to be unknown, we are given an extended vector of ${M}(x)$, referred to as the \emph{exaggerated dictionary}, which includes all nonlinear terms relevant to the system dynamics along with additional \emph{redundant terms} to ensure comprehensive representation of unknown systems (cf. the first nonlinear case study). If the unknown system is polynomial, informed by physical insights, an \emph{upper bound on the maximum degree} of ${M}(x)$ suffices to construct it, capturing all possible combinations of states up to that degree (cf. the second nonlinear case study).

We now present the following lemma, to obtain a data-based representation of closed-loop dt-NS in~\eqref{eq:dt-NS}.

\begin{lemma}
\label{lem:nonlinear-Q-matrix}
    Consider a matrix $Q\in\reals^{T\times N}$ such that
    \begin{equation}
    \label{eq:nonlinear-theta=NQ}
        \identity_N = \mathcal{M}_{0}Q,
    \end{equation}
    where $Q=[Q_1~Q_2]$, $Q_1\in\reals^{T\times n}$ and $Q_2\in\reals^{T\times (N-n)}$, and with  
    \begin{equation}
        \label{eq:data-N0T}
        \mathcal{M}_{0} = \begin{bmatrix}
            &x(0), &x(1), &\ldots, &x(T-1)\\
            &\mathcal{Z}(x(0)), &\mathcal{Z}(x(1)), &\ldots, &\mathcal{Z}(x(T-1))
        \end{bmatrix}
    \end{equation}  
    being an $(N\times T)$ full row-rank matrix. Then the controller designed as $u=K{M}(x)= \mathcal{U}_{0}Q{M}(x)$ leads to the data-based representation of the closed-loop dynamics as
\begin{equation}\label{eq:nonlinearA+BK=X1TQ}
        x^+ \!=\! \mathcal{X}_{1}Q_1x + \mathcal{X}_{1}Q_2\mathcal{Z}(x).
    \end{equation}
\end{lemma}

\begin{proof}
By considering the controller as $u=K{M}(x)=\mathcal{U}_{0}Q{M}(x)$, and according to \eqref{eq:nonlinear-theta=NQ}, we have
\begin{align*}
   x^+ \!= \! (A+BK){M}(x) &=\begin{bmatrix}
        B & A
    \end{bmatrix} \!\!\begin{bmatrix}
        K \\ \identity_N
    \end{bmatrix} \!\!{M}(x)\overset{\eqref{eq:nonlinear-theta=NQ}}{=}
    \begin{bmatrix}
        B & A
    \end{bmatrix} \!\!\begin{bmatrix}
        \mathcal{U}_{0} \\ \mathcal{M}_{0}
    \end{bmatrix} \!\!Q{M}(x)\\
    &=\underbrace{\begin{bmatrix}
        B & A
    \end{bmatrix} \!\!\begin{bmatrix}
        \mathcal{U}_{0} \\ \mathcal{M}_{0}
    \end{bmatrix}}_ {\mathcal{X}_{1}}\!\!\underbrace{\begin{bmatrix}
        Q_1~Q_2
    \end{bmatrix}}_Q\underbrace{\begin{bmatrix}
        x\\
        \mathcal{Z}(x)
    \end{bmatrix}}_{{M}(x)}= \mathcal{X}_{1}Q_1x + \mathcal{X}_{1}Q_2\mathcal{Z}(x),
\end{align*}
with $\mathcal{U}_{0}$ and $\mathcal{M}_{0}$ as in~\eqref{eq:input-data} and~\eqref{eq:data-N0T}, respectively, thereby concluding the proof.
\end{proof}

\begin{remark}
	To ensure that $\mathcal M_0$ in \eqref{eq:nonlinear-theta=NQ} is a full row-rank matrix, the data horizon $T$ must be greater than $N$. This condition is easily verifiable, as $\mathcal M_0$ is constructed from collected data.
\end{remark}

Lemma~\ref{lem:nonlinear-Q-matrix} establishes the dt-NS based on data comprising both linear and nonlinear component. We now consider the controller synthesis technique of nonlinearity cancellation~\cite{de2023learning2} and demonstrate its applicability in the following lemma in computing $k$ evolution steps of the system.

\begin{lemma}
\label{lem:k-evolutions-nonlinear}
    Consider the unknown dt-NS in~\eqref{eq:dt-NS}. The data-based evolution of the system after $k$ steps is acquired as
    \begin{equation}
    	\label{eq:nonlinear-inductive-equation}
    	x^{k+} = (\mathcal{X}_{1}Q_1)^kx,
    \end{equation} under nonlinearity cancellation
    \begin{equation}
    \label{eq:nonlinearity-cancellation}
        \mathcal{X}_1Q_2 = \mathbf{0}_{n\times (N-n)}.
    \end{equation}
\end{lemma}\begin{proof}
   Since $x^{k+}$ is obtained by applying $k$ iterations of $A M(x)+Bu$, it follows that
    \begin{align*}
    	x^+ &= \mathcal{X}_{1}Q_1x + \mathcal{X}_{1}Q_2\mathcal{Z}(x)\stackrel{\eqref{eq:nonlinearity-cancellation}} {=} \mathcal{X}_{1}Q_1x, \\
    	x^{2+} &= \mathcal{X}_{1}Q_1x^+ = (\mathcal{X}_{1}Q_1)^2x,
    \end{align*}
    and by induction, we have
    \begin{equation*}
    	x^{k+} = (\mathcal{X}_{1}Q_1)^kx,
    \end{equation*}
    which concludes the proof.
\end{proof}

\begin{remark}
	Condition~\eqref{eq:nonlinearity-cancellation}  is readily satisfied in scenarios where the nonlinear terms are in the same channel as the control input $u$ (i.e., matched nonlinearities)—see nonlinear case studies.
\end{remark}

Within the same structure of the $k$-CBC as $\mathcal{B}(x) = x^\top Px$, we now leverage Lemma \ref{lem:nonlinear-Q-matrix} and Lemma \ref{lem:k-evolutions-nonlinear} to present our main result in this section, which provides data-driven safety guarantees for unknown dt-NS.

\begin{theorem}
    \label{thm:nonlinearsafe_k-inductive-CBC}
    Consider an unknown dt-NS in~\eqref{eq:dt-NS}, with its data-based representation in Lemma \ref{lem:nonlinear-Q-matrix} and Lemma \ref{lem:k-evolutions-nonlinear}. Suppose there exist a symmetric positive-definite matrix $P\in\reals^{n\times n}$, $k\in\nats_{>0}$, $\gamma, \lambda, \epsilon\in\reals_{\geq0}$, with $\lambda>\gamma + (k-1)\epsilon$, such that the following conditions are satisfied:
    \begin{subequations}
\begin{align}
    \label{eq:nonlinear-safe_kCBC_data1}
        x^\top Px \leq \gamma, &\quad\quad\forall x\in X_{\mathcal I}, \\
    \label{eq:nonlinear-safe_kCBC_data2}
    x^\top Px \geq \lambda, &\quad\quad\forall x\in X_{\mathcal U},\\
    \label{eq:nonlinear-safe_kCBC_data3}
    x^\top \!Q_1^\top \!\mathcal{X}_{1}^\top \!P\mathcal{X}_{1}Q_1x &\leq x^\top \!Px \!+\!\epsilon, ~~ \forall x\in X,\\
    \label{eq:nonlinear-safe_kCBC_data4}
    ((\mathcal{X}_{1}Q_1)^k)^\top P(\mathcal{X}_{1}Q_1)^k &\leq P.
    \end{align}
    \end{subequations}
    Then $\mathcal{B}(x) = x^\top Px$ is a $k$-CBC and $u = \mathcal{U}_{0}Q{M}(x)$ is its corresponding safety controller for the unknown dt-NS.
\end{theorem}
\begin{proof}
    Since $\mathcal{B}(x) = x^\top Px$, it is clear that condition~\eqref{eq:nonlinear-safe_kCBC_data1}-\eqref{eq:nonlinear-safe_kCBC_data2} implies~\eqref{eq:safe_kCBC_cond1}-\eqref{eq:safe_kCBC_cond2}, respectively. We now proceed to show condition~\eqref{eq:safe_kCBC_cond3}. Using the results of Lemmas~\ref{lem:nonlinear-Q-matrix} and \ref{lem:k-evolutions-nonlinear}, we have
    \begin{align*}
    \mathcal{B}(x^+)
    &= (A{M}(x)+Bu)^\top P(A{M}(x)+Bu) \\
    &\stackrel{\eqref{eq:nonlinearA+BK=X1TQ}} {=} (\mathcal{X}_{1}Q_1x+\!\!\!\!\!\underbrace{\mathcal{X}_1Q_2}_{\mathbf{0}_{n\times (N-n)}}\!\!\!\!\!\mathcal{Z}(x))^\top P(\mathcal{X}_{1}Q_1x+\!\!\!\!\!\underbrace{\mathcal{X}_1Q_2}_{\mathbf{0}_{n\times (N-n)}}\!\!\!\!\!\mathcal{Z}(x))\\ 
    &\stackrel{\eqref{eq:nonlinearity-cancellation}} {=} (\mathcal{X}_{1}Q_1x)^\top P(\mathcal{X}_{1}Q_1x)=x^\top Q_1^\top\mathcal{X}_{1}^\top P\mathcal{X}_{1}Q_1x.
    \end{align*}
    By leveraging condition~\eqref{eq:nonlinear-safe_kCBC_data3}, one has 
    \begin{align*}
    	\mathcal{B}(x^{+}) = x^\top Q_1^\top\mathcal{X}_{1}^\top P\mathcal{X}_{1}Q_1x\leq \underbrace{x^\top P x}_{\mathcal{B}(x)}  + \epsilon.
    \end{align*}
    As the final step of the proof, we show that condition~\eqref{eq:safe_kCBC_cond4} holds true, as well. From \eqref{eq:nonlinear-inductive-equation}, one has
    \begin{align*}
         \mathcal{B}(x^{k+}) &= ((\mathcal{X}_{1}Q_1)^kx)^\top P((\mathcal{X}_{1}Q_1)^kx)\\
          &= x^\top ((\mathcal{X}_{1}Q_1)^k)^\top P((\mathcal{X}_{1}Q_1)^k)x\overset{\eqref{eq:nonlinear-safe_kCBC_data4}}{\leq} \underbrace{x^\top P x}_{\mathcal{B}(x)}.
    \end{align*}
    Consequently, $\mathcal{B}(x) = x^\top Px$
    is a $k$-CBC and $u=K{M}(x)=\mathcal{U}_{0}Q{M}(x)$ is its corresponding safety controller for the unknown dt-NS, which completes the proof.
\end{proof}

As a direct consequence of Theorem~\ref{thm:nonlinearsafe_k-inductive-CBC}, the following corollary highlights the special case where $k=1$, reducing the problem to the CBC (if it exists), as described in Definition~\ref{def:safe_barrier-conditions}.

\begin{corollary}
    \label{thm:nonlinearsafe_CBC}
    Consider an unknown dt-NS $x^+ = A{M}(x)+Bu$, with its data-based representation $x^+=\mathcal{X}_{1}Q_1x$ under nonlinearity cancellation~\eqref{eq:nonlinearity-cancellation}. Suppose there exist a symmetric positive-definite matrix $P\in\reals^{n\times n}$ and a matrix $H\in \reals^{T\times n}$, with $Q_1 = HP$, and constants $\gamma, \lambda\in\reals_{\geq0}$, with $\lambda>\gamma$, such that the following conditions are satisfied:
    \begin{subequations}
\begin{align}\label{New7}
    & \mathcal M_0Q_2 = \begin{bmatrix}
    	\mathbf{0}_{n\times (N-n)}\\
    	\identity_{N-n}
    \end{bmatrix}\!\!,\\\label{New7-1}
    	&\mathcal M_0H = \begin{bmatrix}
    	P^{-1}\\
    	\mathbf{0}_{(N-n)\times n}
    \end{bmatrix}\!\!,\\
    \label{eq:nonlinear-safe_CBC_data1}
        &x^\top Px \leq \gamma, \quad\quad\quad\quad\forall x\in X_{\mathcal I}, \\
    \label{eq:nonlinear-safe_CBC_data2}
    &x^\top Px \geq \lambda, \quad\quad\quad\quad\forall x\in X_{\mathcal U},\\
    \label{eq:nonlinear-safe_CBC_data3}
    &\begin{bmatrix}
    	P^{-1} & \mathcal{X}_1H  \\
    	H^\top\mathcal{X}_1^\top & P^{-1}
    \end{bmatrix}\geq 0.
    \end{align}
    \end{subequations}
    Then $\mathcal{B}(x) = x^\top Px$ is a conventional CBC (with $k = 1$) and $u = K {M}(x)= \mathcal{U}_{0}Q{M}(x)$ is its corresponding safety controller for the unknown dt-NS.
\end{corollary}
    
   \begin{proof}
   	According to $\mathcal{B}(x) = x^\top Px$, it is clear that condition~\eqref{eq:nonlinear-safe_CBC_data1}-\eqref{eq:nonlinear-safe_CBC_data2} implies~\eqref{eq:safe_kCBC_cond1}-\eqref{eq:safe_kCBC_cond2}, respectively.
   	
   	By defining $Q_1=H P$, it follows directly that condition~\eqref{New7} along with~\eqref{New7-1} leads to
   	\begin{align}\label{new8}
   		\mathcal M_0\underbrace{[Q_1 ~~~ Q_2]}_Q= \identity_N,
   	\end{align}
   	which is in compliance with condition \eqref{eq:nonlinear-theta=NQ}. 
   	
   	 We now have
   	\begin{align}\notag
   		\mathcal{B}(x^+)
   		&= (A{M}(x)+Bu)^\top P(A{M}(x)+Bu) \\\notag
   		&\stackrel{\eqref{eq:nonlinearA+BK=X1TQ}} {=} (\mathcal{X}_{1}Q_1x+\!\!\!\!\!\underbrace{\mathcal{X}_1Q_2}_{\mathbf{0}_{n\times (N-n)}}\!\!\!\!\!\mathcal{Z}(x))^\top P(\mathcal{X}_{1}Q_1x+\!\!\!\!\!\underbrace{\mathcal{X}_1Q_2}_{\mathbf{0}_{n\times (N-n)}}\!\!\!\!\!\mathcal{Z}(x))\\ \notag
   		&\stackrel{\eqref{eq:nonlinearity-cancellation}} {=} (\mathcal{X}_{1}Q_1x)^\top P(\mathcal{X}_{1}Q_1x)=x^\top Q_1^\top\mathcal{X}_{1}^\top P\mathcal{X}_{1}Q_1x\\\label{New88}
   		&=x^\top \underbrace{P^\top P^{-1}}_{\identity_n}Q_1^\top \mathcal{X}_{1}^\top P\mathcal{X}_{1}Q_1\underbrace{P^{-1}P}_{\identity_n}x=x^\top P^\top H^\top \mathcal{X}_{1}^\top P\mathcal{X}_{1}HPx,
   	\end{align}
   	  with $Q_1P^{-1} = H$ (since $Q_1=H P$). According to the Schur complement \cite{zhang2006schur}, and considering  condition \eqref{eq:nonlinear-safe_CBC_data3}, it is well-established that
   	\begin{align*}
   		\begin{bmatrix}
   			P^{-1} & \mathcal{X}_1H  \\
   			H^\top\mathcal{X}_1^\top & P^{-1}
   		\end{bmatrix}\geq 0& \Leftrightarrow P^{-1} - H^\top \!\!\mathcal{X}_{1}^\top \!\!P\mathcal{X}_{1}H \geq 0.
   	\end{align*}
   	Thus, it is clear that
   	\begin{align}\label{New87}
   		        x^\top P^\top H^\top \mathcal{X}_{1}^\top P\mathcal{X}_{1}Q_1HPx \leq \overbrace{x^\top  \underbrace{P^\top P^{-1}}_{\identity_n} Px}^{\mathcal B(x)}.
   	\end{align}
   	Be considering~\eqref{New88},\eqref{New87}, one can deduce that 
   	$$\mathcal{B}(x^+) \leq \mathcal{B}(x),$$
   	which concludes the proof.
   \end{proof}

\section{Computation of $k$-CBC via SOS Optimization}
\label{sec:SOS}

Computing $k$-CBCs is generally a challenging task, even within a model-based framework. Here, we provide a sum-of-squares (SOS) optimization approach to derive $k$-CBC together with their corresponding safety controllers from data, assuming the regions of interest $X$, $X_{\mathcal I}$, and $X_{\mathcal U}$ are semi-algebraic~\cite{bochnak2013real}. Specifically, a semi-algebraic set $X \subseteq \mathbb{R}^n$ can be described by a vector of polynomials $g(x)$, \emph{i.e.,} $X = \{x \in \mathbb{R}^n \mid g(x) \geq 0\}$, with the inequalities applied element-wise.

To achieve this, we reformulate conditions~\eqref{eq:nonlinear-safe_kCBC_data1}-\eqref{eq:nonlinear-safe_kCBC_data3} into a sum-of-squares (SOS) optimization program~\cite{parrilo2003semidefinite}, which can be solved using \textsf{SOSTOOLS} \cite{papachristodoulou2013sostools}. Meanwhile, condition \eqref{eq:nonlinear-safe_kCBC_data4} can be addressed using a semidefinite programming (SDP) solver such as \textsf{SeDuMi}~\cite{sturm1999using}.

\begin{lemma}
	\label{lem:nonlinear-safe-SOS-constraints}
	Consider a dt-NS $\Sigma$, with the sets $X$, $X_{\mathcal I}$, and $X_{\mathcal U}$ being represented by vectors of polynomials $g(x)$, $g_{\mathcal I}(x)$, and $g_{\mathcal U}(x)$, respectively. Suppose there exist an SOS polynomial $x^\top Px$ with a symmetric positive-definite matrix $P\in\reals^{n\times n}$, constants $k\in\nats_{>0}$, $\gamma,\lambda,\epsilon\in\reals_{\geq0}$, with $\lambda>\gamma + (k-1)\epsilon$, and vectors of SOS polynomials $L(x), L_{\mathcal I}(x),$ and $L_{\mathcal U}(x)$ of appropriate dimensions such that 
	\begin{subequations}
		\begin{align}
			\label{eq:nonlinear_safe_SOS_cond1}
			-&x^\top Px - L_{\mathcal I}^\top(x)g_{\mathcal I}(x) + \gamma, \\
			\label{eq:nonlinear_safe_SOS_cond2}
			&x^\top Px - L_{\mathcal U}^\top(x)g_{\mathcal U}(x) - \lambda, \\
			\label{eq:nonlinear_safe_SOS_cond3}
			-&x^\top Q_1^\top \mathcal{X}_{1}^\top P\mathcal{X}_{1}Q_1x + x^\top Px - L^\top(x)g(x) + \epsilon,
		\end{align}
	\end{subequations}
    are SOS polynomials, while  condition~\eqref{eq:nonlinear-safe_kCBC_data4} is also satisfied.
	Then the function $x^\top Px$ is a $k$-CBC as in Theorem~\ref{thm:nonlinearsafe_k-inductive-CBC} satisfying conditions~\eqref{eq:nonlinear-safe_kCBC_data1}-\eqref{eq:nonlinear-safe_kCBC_data4}, and $u = \mathcal{U}_{0}Q{M}(x)$ is its corresponding safety controller for the unknown dt-NS.
\end{lemma}

The proof of this lemma naturally extends from the existing works in the literature, including~\cite{prajna2004safety, Anand2021kSafety}.

\begin{algorithm}[t]
	\caption{Design of $k$-CBC and safety controller for dt-NS}\label{alg:nonlinear}
	\begin{algorithmic}[1]
	\Require Regions of interest $X,X_{\mathcal I},X_{\mathcal U}$, extended dictionary ${M}(x)$, relaxation horizon $k$
	\State Collect trajectories $\mathcal{U}_0, \mathcal{X}_0,\mathcal{X}_1$, according to~\eqref{data}
	\State Construct the \emph{full row-rank} matrix $\mathcal{M}_0$ in~\eqref{eq:data-N0T} from data
    \State Solve\footnotemark\!~\eqref{eq:nonlinearity-cancellation}, \eqref{New7}, \eqref{New7-1}, and~\eqref{eq:nonlinear-safe_CBC_data3} for designing $P$, $H$, and $Q$ (with partitions $Q_1$ and $Q_2$)
    \If {Feasible $\lambda$ and $\gamma$ (with $\lambda>\gamma$) is found by satisfying~\eqref{eq:nonlinear-safe_CBC_data1} and~\eqref{eq:nonlinear-safe_CBC_data2}}
    \Statex\textbf{Return:} Conventional CBC $\mathcal{B}(x) = x^\top Px$ and safety controller $u = \mathcal{U}_{0}Q{M}(x)$ 
    \ElsIf{Discard $P$ and potentially find new symmetric positive-definite $P$ and $\lambda,\gamma,\epsilon$ with $\lambda>\gamma + (k-1)\epsilon$ satisfying~\eqref{eq:nonlinear_safe_SOS_cond1}-\eqref{eq:nonlinear_safe_SOS_cond3}  and~\eqref{eq:nonlinear-safe_kCBC_data4}}
    \Statex \textbf{Return:} $k$-CBC $\mathcal{B}(x) = x^\top Px$ and safety controller $u = \mathcal{U}_{0}Q{M}(x)$ 
    \Else
        \Statex\textbf{Return:} {Feasible solution not found with given $k$}
    \EndIf
	\end{algorithmic}
\end{algorithm}
\footnotetext{We define $\mathcal{E} = P^{-1}$ in conditions~\eqref{New7-1} and~\eqref{eq:nonlinear-safe_CBC_data3}, while ensuring it is a \emph{symmetric positive-definite} matrix. Once these conditions are satisfied and $\mathcal{E}$ is designed, the matrix $P$ can be derived as $\mathcal{E}^{-1} = (P^{-1})^{-1} = P$.}

\begin{remark}
	To accommodate any \emph{arbitrary number} of unsafe regions $X_{\mathcal{U}_i}$, where $i\in\{1,\ldots, z\}$, condition \eqref{eq:nonlinear_safe_SOS_cond2} should be reiterated and enforced for each distinct unsafe region.
\end{remark}

\begin{remark}
	Using similar reasoning as in Lemma \ref{lem:nonlinear-safe-SOS-constraints}, conditions \eqref{eq:safe_kCBC_data1}-\eqref{eq:safe_kCBC_data3} can be reformulated as an SOS optimization problem for dt-LS.
\end{remark}

We present Algorithm~\ref{alg:nonlinear} that provides the required steps for finding the $k$-CBC and its safety controller for dt-NS. As a general control strategy, one can first aim to find a conventional CBC (with $k=1$) and synthesize its controller via the conditions of Corollary~\ref{thm:nonlinearsafe_CBC}. If a feasible CBC is not found, one can utilize the pre-synthesized controller for $k=1$ with the matrix gain $\mathcal{U}_{0}Q$ and attempt to find a relaxed $k$-CBC, according to Theorem~\ref{thm:nonlinearsafe_k-inductive-CBC}, by designing new $P$ and level sets $\gamma$ and $\lambda$.

\section{Case Studies}
\label{sec:case-studies}

We employ a series of physical benchmarks with unknown dynamics to showcase the efficacy of our data-driven approach for the construction of $k$-inductive CBCs and their corresponding safety controllers.

\subsection{Linear Systems}

{\bf RLC Circuit.} Consider the following RLC circuit~\cite{Anand2021kSafety}
\begin{align*}
    \Sigma : \begin{cases}
        x_1^+ = x_1 + \tau(-\frac{R}{L}x_1 - \frac{1}{L}x_2) + u_1, \\
        x_2^+ = x_2 + \tau(\frac{1}{C}x_1) + u_2,
    \end{cases}
    \end{align*}
    where $x_1$ and $x_2$ denote the current and the voltage of the circuit, respectively, $\tau=0.5$ is the sampling time, $R=2$ is the series resistance, $L=9$ is the series inductance, and $C=0.5$ is the capacitance of the circuit. The dynamics of the RLC circuit are in the form of~\eqref{eq:dt-LS} with
    \begin{align*}
        A=\begin{bmatrix}
            1 \!-\! \tau\frac{R}{L} && -\frac{1}{L}\\
            \tau\frac{1}{C} && 1
        \end{bmatrix}\!\!, \quad B = \begin{bmatrix}
            1 && 0\\
            0 && 1
        \end{bmatrix}\!\!.
    \end{align*}
    We assume both matrices $A$ and $B$ are \emph{unknown} to us.

        \begin{figure}[t]
    \centering
    \includegraphics[width=0.47\columnwidth]{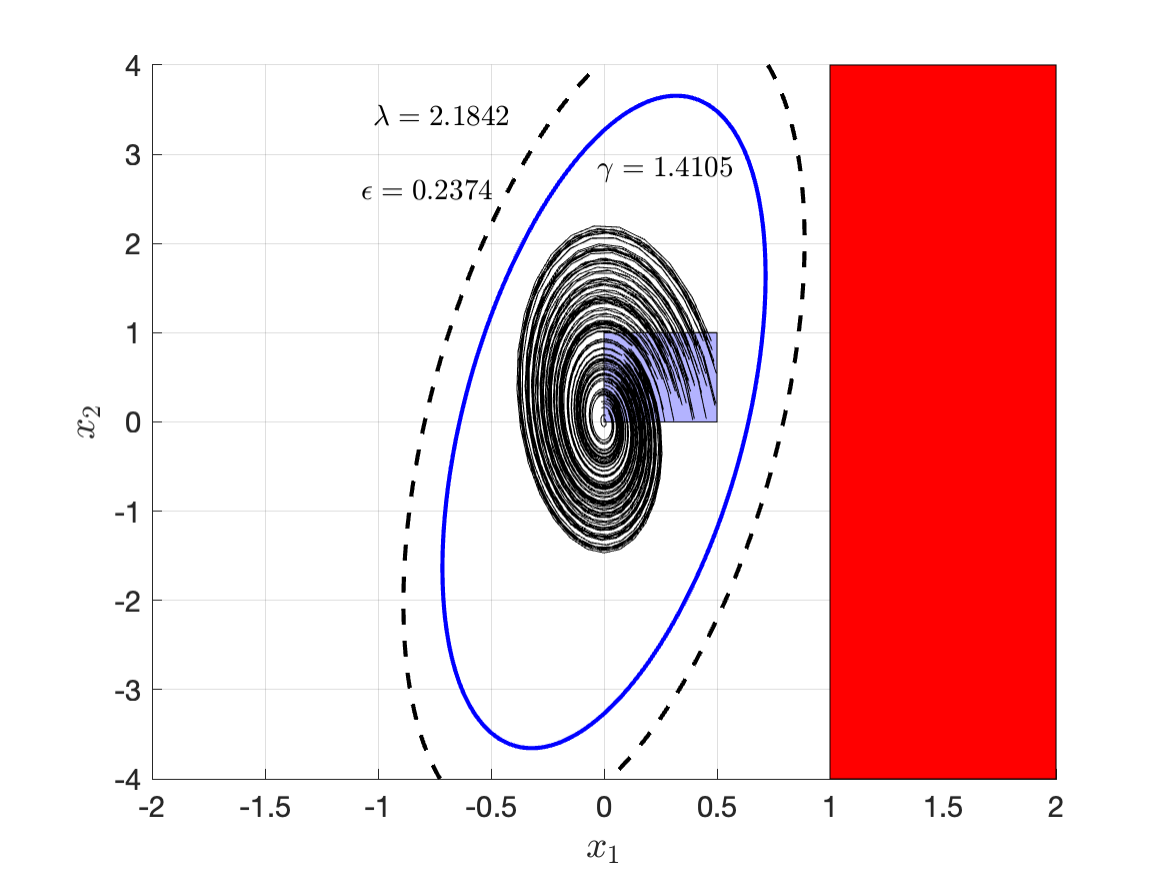}
    \caption{Closed-loop state trajectories of the RLC circuit ($100$ trajectories shown in black) under the designed controller and $k$-CBC with $k=3$, as defined in~\eqref{eq:RLC_controller}. Purple and red boxes are initial and unsafe regions, respectively. The level sets $\gamma$ and $\lambda$ are depicted in blue and dashed black, respectively.}
    \label{fig:safe-RLC}
\end{figure}
    
     The state space of the system is $X=[-2,2]\times[-4,4]$, with initial set $X_{\mathcal I}=[0,0.5]\times[0,1]$, and unsafe set $X_{\mathcal U}=[1,2]\times[-4,4]$. We consider Theorem~\ref{thm:safe_k-inductive-CBC} to construct a $k$-CBC of the form $x^\top Px$. Following Algorithm~\ref{alg:linear}, we first collect trajectories $\mathcal{U}_{0},\mathcal{X}_{0}$ and $\mathcal{X}_{1}$ with horizon $T=30$. After solving\footnote{Matrix $Q$ is not reported due to its large size of $(30 \times 2$).} $Q$ according to~\eqref{eq:Q-matrix-X0}, we satisfy conditions~\eqref{eq:safe_kCBC_data1}-\eqref{eq:safe_kCBC_data4} for $k=3$ with designed parameters
    \begin{align*}
    	P = {\begin{bmatrix}
    	    3.3600  & -0.2943\\
   -0.2943  &  0.1285
    	\end{bmatrix}}\!,\quad \epsilon=0.2374,\quad \gamma=1.4105, \quad
    	\lambda=2.1842.
    \end{align*}
    Then the $k$-CBC with its corresponding safety controller are designed as 
    \begin{align}\notag
    \mathcal{B}(x) &= 3.36x_1^2 - 0.58856x_1x_2 + 0.1285x_2^2,\\\label{eq:RLC_controller}
        u &= \begin{bmatrix}
            0.024862 &&& 0.0075704\\
            0.078083 &&& -0.02691
        \end{bmatrix}\!\!\begin{bmatrix}
            x_1\\
            x_2
        \end{bmatrix}\!\!.
    \end{align}

Closed-loop state trajectories of the RLC circuit under the designed controller using data are depicted in Fig.~\ref{fig:safe-RLC}. We highlight that without the relaxation of $k$-induction, no solution could be obtained for the RLC circuit when $k=1$.

{\bf DC Motor.} Consider the following DC motor~\cite{adewuyi2013dc}
\begin{align*}
    \Sigma : \begin{cases}
        x_1^+= x_1- \tau(\frac{R}{L}x_1 + \frac{k_{dc}}{L}x_2) + u_1, \\
        x_2^+ = x_2 + \tau(\frac{k_{dc}}{J}x_1 - \frac{b}{J}x_2) + u_2,
    \end{cases}
    \end{align*}
    where $x_1,~x_2,~R=1,~L=0.01,$ and $J=0.01$ are the armature current, the rotational speed of the shaft, the electrical resistance, the electrical inductance, and the moment of inertia of the rotor, respectively. Additionally, $\tau=0.01$, $b=1$, and $k_{dc}=0.01$, represent, respectively, the sampling time, the motor torque, and the back electromotive force. The dynamics of the DC motor are in the form of~\eqref{eq:dt-LS} with 
    \begin{align*}
        A=\begin{bmatrix}
            1 \!-\! \tau\frac{R}{L} && -\tau\frac{k_{dc}}{L}\\
            \tau\frac{k_{dc}}{J} && 1-\tau\frac{b}{J}
        \end{bmatrix}\!\!, \quad B = \begin{bmatrix}
            1 && 0\\
            0 && 1
        \end{bmatrix}\!\!.
    \end{align*}
    We assume both matrices $A$ and $B$ are \emph{unknown}.

    \begin{figure}[t]
    \centering
    \includegraphics[width=0.47\columnwidth]{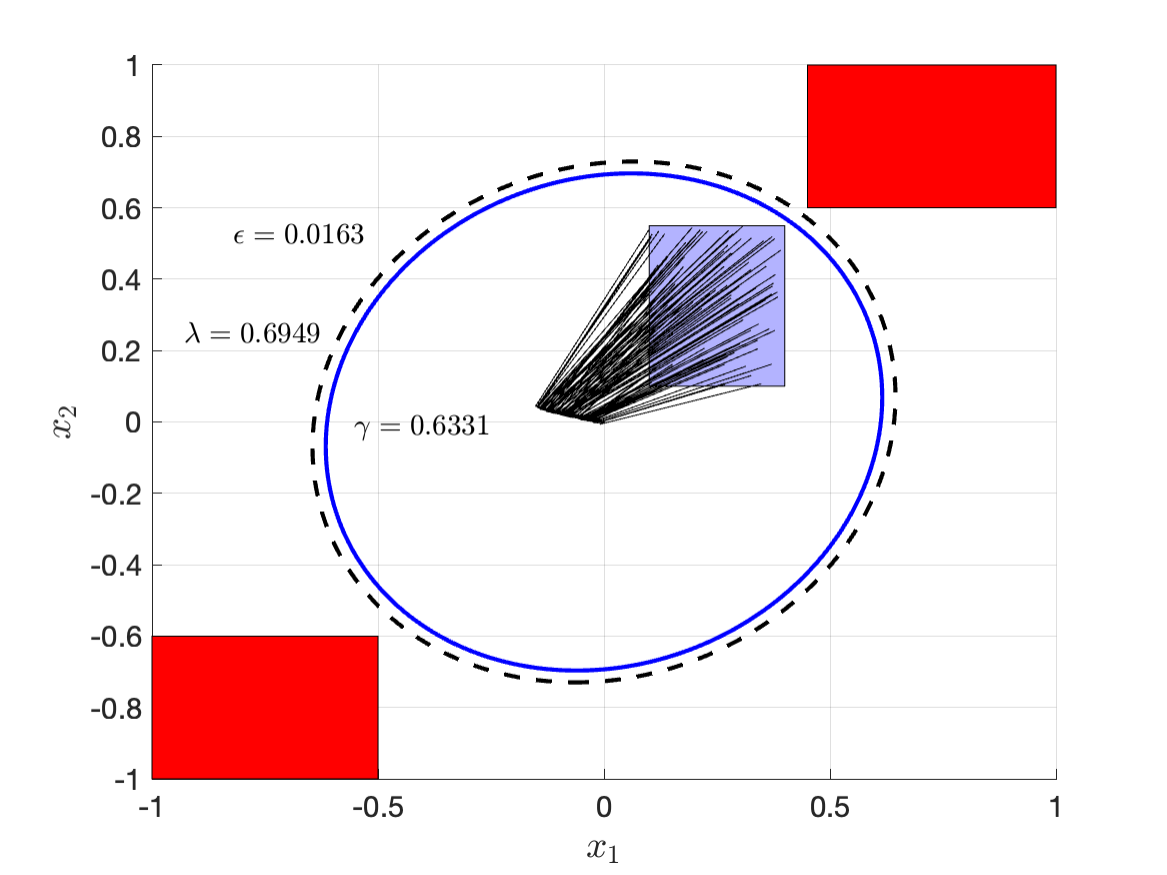}
    \caption{Closed-loop state trajectories of the DC motor ($100$ trajectories shown in black) under the designed controller and $k$-CBC with $k=3$, as defined in~\eqref{eq:DC_controller}. Purple and red boxes are initial and unsafe regions, respectively. The level sets $\gamma$ and $\lambda$ are depicted in blue and dashed black, respectively.}
    \label{fig:safe-DC}
\end{figure}

    The regions of interest are given as $X={[-1,1]^2}$, $X_{\mathcal I}=[0.1,0.4]\times[0.1,0.55]$, and $X_{\mathcal U}={[0.45,1]\times[0.6,1]}\cup{[-1,-0.5]\times[-1,-0.6]}$. We collect trajectories $\mathcal{U}_{0},\mathcal{X}_{0}$ and $\mathcal{X}_{1}$ with horizon $T=30$. We then solve~\eqref{eq:Q-matrix-X0} for\footnote{Matrix $Q$ is not reported due to its large size of $(30 \times 2$).} $Q$ and fulfill conditions~\eqref{eq:safe_kCBC_data1}-\eqref{eq:safe_kCBC_data4} for $k=3$ with designed parameters
    \begin{align*}
    	P = {\begin{bmatrix}
    	    1.6873 &&  -0.1467\\
   -0.1467  &&  1.3181
    	\end{bmatrix}}\!, \quad\epsilon=0.0163, \quad\gamma=0.6331, \quad
    	\lambda=0.6949.
    \end{align*}
    Then the $k$-CBC with its corresponding safety controller are designed as 
        \begin{align}\notag
    	\mathcal{B}(x) &= 1.6873x_1^2 -0.29336x_1x_2 + 1.3181x_2^2,\\\label{eq:DC_controller}
    	u &= \begin{bmatrix}
    		0.063901 & - 0.28251\\
    		-0.05539 & 0.090067
    	\end{bmatrix}\!\!\begin{bmatrix}
    		x_1\\
    		x_2
    	\end{bmatrix}\!\!.
    \end{align}
    
Closed-loop state trajectories of the DC motor under the designed data-driven controller are depicted in Fig.~\ref{fig:safe-DC}.

\subsection{Nonlinear Systems}

{\bf Nonpolynomial Car.} As a nonlinear benchmark, we consider a 3-dimensional car with the following nonpolynomial dynamics~\cite{edwards2024fossil}
\begin{equation*}
    \Sigma : \begin{cases}
        x_1^+ = x_1 + \tau\text{sin}(x_3)+\tau u_1, \\
        x_2^+ = x_2 + \tau\text{cos}(x_3) +\tau u_2,\\
        x_3^+ = (1+ \tau)x_3 +\tau u_3,
    \end{cases}
\end{equation*}
where $\tau= 0.1$ is the sampling time, $x_1$ and $x_2$ are the car's coordinate position, while $x_3$ is the car's heading angle. The dynamics can be written in the form of~\eqref{eq:dt-NS} with 
\begin{align*}
 &A=\begin{bmatrix}
    1 && 0 && 0 && \tau&& 0\\
    0 && 1 && 0&&0&&\tau\\
    0&&0&&1\!+\!\tau&&0&&0
\end{bmatrix}\!\!, \quad B =\begin{bmatrix}
    \tau && 0 && 0\\
    0 && \tau && 0\\
    0&&0&&\tau
\end{bmatrix}\!\!,\\ & M(x) = \begin{bmatrix}
    x_1 &&
    x_2&&
    x_3&&
    \text{sin}(x_3)&&\text{cos}(x_3)
\end{bmatrix}^\top\!\!\!\!.\end{align*}
While  both matrices $A$ and $B$, and the exact vector $M(x)$ are \emph{unknown}, we assume the \emph{extended dictionary} ${M}(x) = [x_1~~x_2~~x_3~~ \text{sin}(x_3)~~\text{cos}(x_3)~~x_1^2~~x_2^2]^\top$ with irrelevant terms $x_1^2,x_2^2$ provided.

\begin{figure}[t]
	\centering
	\includegraphics[width=0.47\columnwidth]{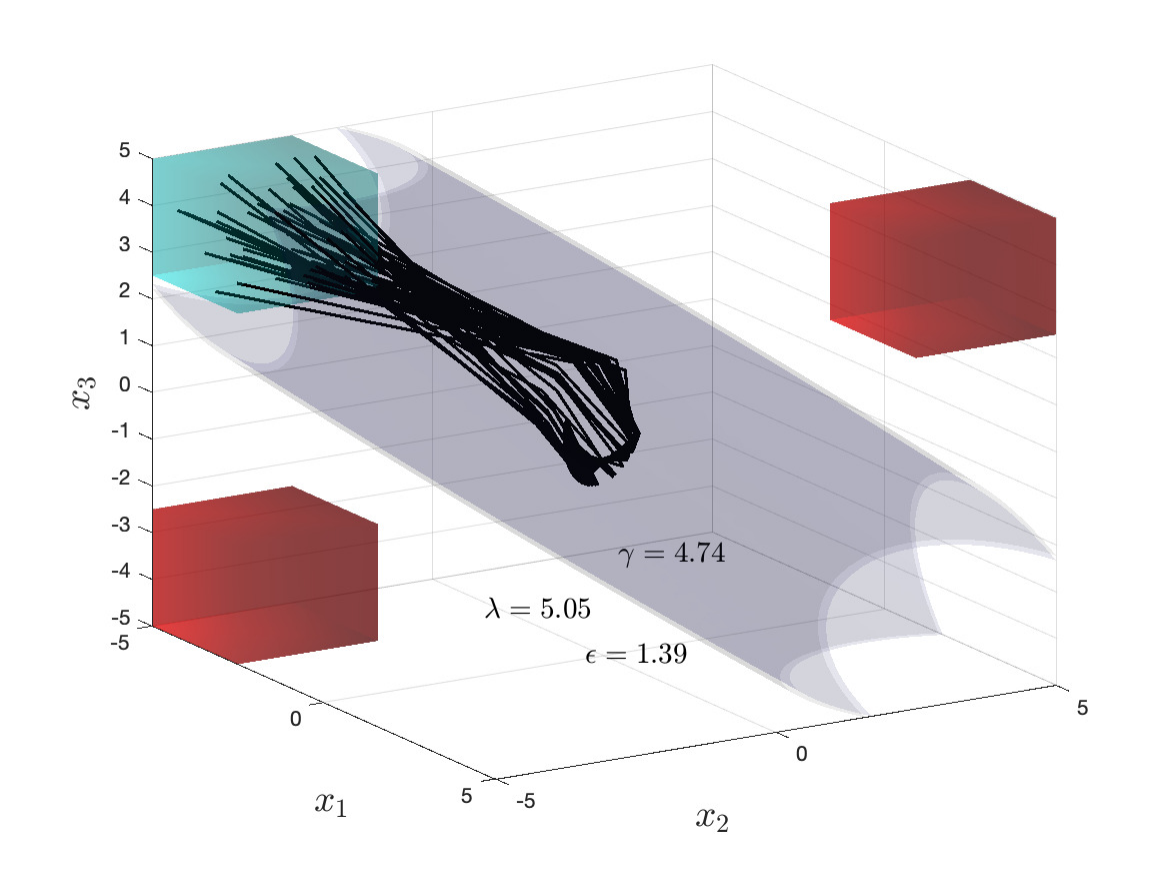}
	\caption{Closed-loop state trajectories of the nonpolynomial car ($50$ trajectories shown in black) under the designed controller and $k$-CBC with $k=2$. Emerald green and maroon boxes are initial and unsafe regions, respectively. The level sets $\gamma$ and $\lambda$ are depicted in dark and light gray, respectively. }
	\label{fig:safe-car}
\end{figure}

The regions of interest are given as: state space $X=[-5, 5]^3$, initial set $X_{\mathcal I}=[-5, -2.5]^2\times[2.5,5]$, and two unsafe sets $X_{\mathcal U}=[-5,-2.5]^3\cup[2.5,5]^3$. 
Following Algorithm~\ref{alg:nonlinear}, we collect data $\mathcal{U}_0,\mathcal{X}_0$ and $\mathcal{X}_1$ with horizon $T=50$. We design the $k$-CBC $\mathcal{B}(x) = x^\top Px$ for $k = 2$ with \[P = \begin{bmatrix}
0.4064 &&  -0.1193  &&  0.2695\\
   -0.1193   && 0.4064  &&  0.2695\\
    0.2695  &&  0.2695  &&  0.5401
\end{bmatrix}\!\!,\] level sets $\gamma=4.7361, 
 \lambda=5.0504$, $\epsilon=1.3869$, and the controller $u =\mathcal{U}_0Q{M}(x)$, where
 \[\mathcal{U}_0Q \!=\!\!  \begin{bmatrix}
-10.0 && 0.022 && -0.0002 \\
-0.01 && -10.0 && 0.0001 \\
-0.02 && -0.004 && -11.0 \\
3.29 && 3.51 && -1.50 \\
-5.02 && -0.30 && -14.8 \\
-0.17 && 0.04 && 0.01 \\
-0.004 && -0.23 && 0.008
\end{bmatrix}^\top\!\!\!\!\!\!.\]

Closed-loop state trajectories of the nonpolynomial car  under the designed data-driven controller are depicted in Fig.~\ref{fig:safe-car}. Note that without the relaxation provided by $k$-induction, no solution could be obtained for the nonpolynomial car when $k = 1$. 

{\bf Lorenz Attractor.} As the last benchmark, we consider the Lorenz attractor, a well-studied dynamical system with chaotic behavior, with the following dynamics~\cite{abate2020arch}
\begin{equation*}
    \Sigma : \begin{cases}
        x_1^+= x_1 + \tau\sigma (x_1 + x_2+u_1), \\
        x_2^+ = x_2 +\tau\rho x_1 - \tau x_2-\tau x_1x_3 + \tau u_2, \\
        x_3^+ = x_3 +\tau (x_1x_2-\beta x_3+u_3),
    \end{cases}
\end{equation*}
where $\rho=28$ is a coupling parameter for $x_1$ and $x_2$, $\sigma=10$ is a scaling parameter of $x_1$, $\beta=\frac{8}{3}$ is the damping coefficient for $x_3$, and $\tau=10^{-3}$ is the sampling time. The dynamics of the Lorenz system are in the form of~\eqref{eq:dt-NS} with
\begin{align*}
 &A=\begin{bmatrix}
    1\!+\!\tau\sigma && \tau\sigma && 0 && 0&& 0\\
    \tau\rho && 1\!-\!\tau && 0&&0&&-\tau\\
    0&&0&&1\!-\!\tau\beta&&\tau&&0
\end{bmatrix}\!\!, \quad B =\begin{bmatrix}
    \tau\sigma && 0 && 0\\
    0 && \tau && 0\\
    0&&0&&\tau
\end{bmatrix}\!\!,  \\ & M(x) = \begin{bmatrix}
    x_1 &&
    x_2&&
    x_3&&
    x_1x_2&&
    x_1x_3
\end{bmatrix}^\top\!\!\!\!.\end{align*}
While  both matrices $A$ and $B$, and the exact vector $M(x)$ are \emph{unknown}, we assume that only the \emph{maximum degree} of $2$ is provided for ${M}(x)$. Based on this, we construct an extended dictionary for ${M}(x)$ to include all possible combinations of monomials up to degree  $2$ as
\begin{align*}
	{M}(x) =[x_1~~~x_2~~~x_3~~~x_1x_2~~~ x_1x_3~~~ x_2x_3~~~x_1^2~~x_2^2~~x_3^2]^\top\!\!.
\end{align*}

\begin{figure}[t]
	\centering
	\includegraphics[width=0.47\columnwidth]{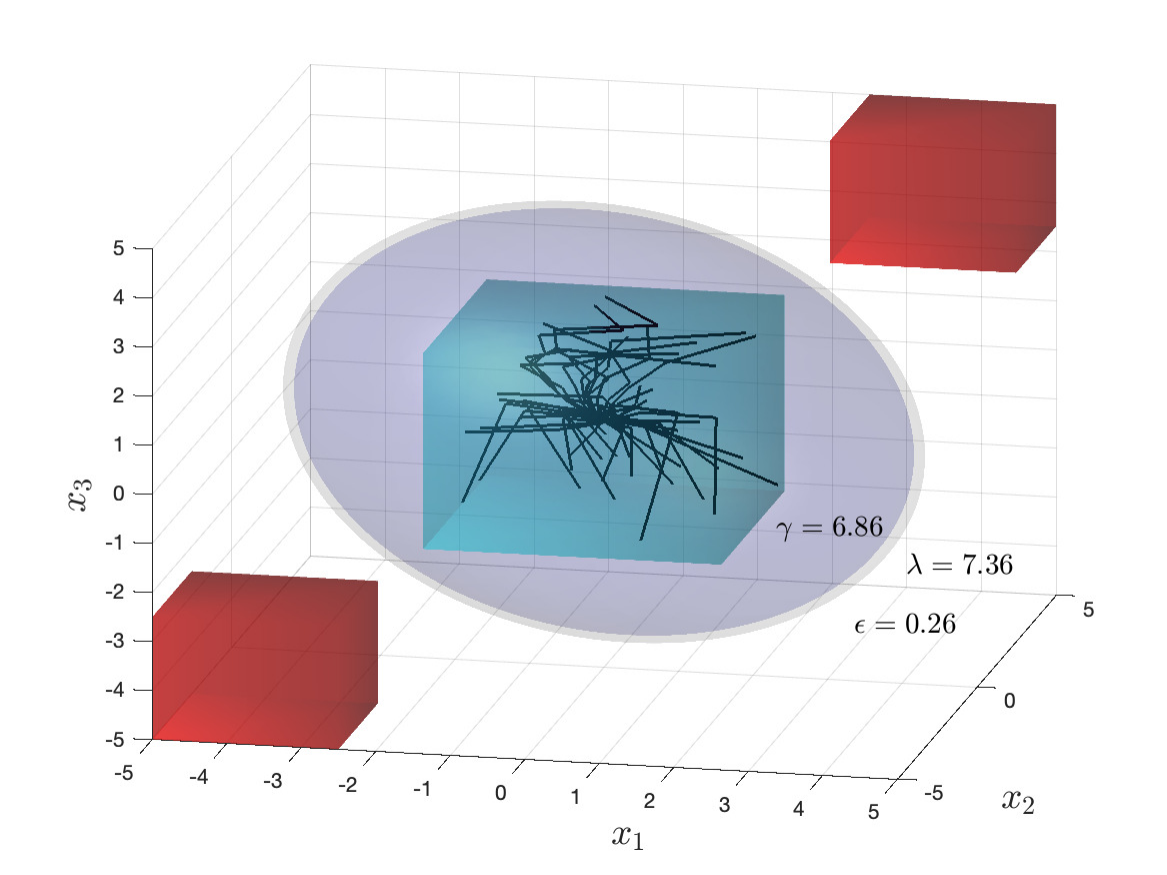}
	\caption{Closed-loop state trajectories of the Lorenz system ($50$ trajectories shown in black) under the designed controller and $k$-CBC with $k=2$. Emerald green and maroon boxes are initial and unsafe regions, respectively. The level sets $\gamma$ and $\lambda$ are depicted in dark and light gray, respectively.}
	\label{fig:lorenz}
\end{figure}

The regions of interest are given as follows: state space $X=[-5,5]^3$, initial set $X_{\mathcal I}=[-2, 2]^3$ and two unsafe sets $X_{\mathcal U}=[-5,-2.5]^3\cup[2.5,5]^3$.  We collect data $\mathcal{U}_0,\mathcal{X}_0$, $\mathcal{X}_1$ with horizon $T=50$ and design the $k$-CBC $\mathcal{B}(x) = x^\top Px$ for $k = 2$ with
\begin{align*}
	P =\begin{bmatrix}
		0.4068 &&    0.0457 &&   0.0469\\
		0.0457  &&  0.4108   && 0.0470\\
		0.0469   && 0.0470   && 0.4032\\
	\end{bmatrix}\!\!,
\end{align*} 
level sets $\gamma=6.8601, 
\lambda=7.3572$, $\epsilon=0.2593$, and the controller $u =\mathcal{U}_0Q{M}(x)$, where
\begin{align*}\label{eq:Lorenz_controller}
	\mathcal{U}_0Q = \begin{bmatrix}
		-688.1  &&  61.3   && -6.26 \\
		70.1    && -697.7  &&  16.6 \\
		11.1    && -5.48   && -658.2 \\
		-33.4   && -70.8   && -39.6 \\
		315.9   && -100.7  &&  79.5 \\
		-151.3  &&  219.2  && -74.1 \\
		58.6    &&  30.0   && -4.04 \\
		18.2    &&  75.5   &&  16.0 \\
		-113.0  && -128.1  &&  295.7
	\end{bmatrix}^\top\!\!\!\!\!\!.
	\end{align*}

Closed-loop state trajectories of the Lorenz system under the designed data-driven controller are depicted in Fig.~\ref{fig:lorenz}. 

\section{Conclusion}\label{sec:conclusion}
In this work, we developed safety controllers for discrete-time nonlinear systems with \emph{unknown} dynamics using $k$-inductive control barrier certificates ($k$-CBCs). Our relaxed method enhances the chances of constructing feasible CBCs and establishing safety guarantees for such systems. Given that the system dynamics are not explicitly known, we leveraged the concept of persistency of excitation, which ensures that input-state data from a single trajectory can effectively capture the system’s behavior, provided that the data meets a specific rank condition. We employed sum-of-squares optimization to synthesize both the $k$-CBC and its safety controller directly from such a single trajectory, ensuring the safe operation of the unknown system. We validated the effectiveness of our approach through a series of physical benchmarks with unknown dynamics. Expanding our data-driven framework to \emph{continuous-time} nonlinear systems using \emph{$t$-barrier certificates} is a subject of ongoing investigation for future work.

\bibliographystyle{alpha}
\bibliography{biblio}

\end{document}